\documentclass[12pt]{amsart}
\usepackage{graphicx}
\usepackage{a4wide}
\usepackage[centertags]{amsmath}
\usepackage{amsfonts}
\usepackage{amssymb}
\usepackage{amsthm}
\usepackage{hyperref}
\newtheorem{thm}{Theorem}%[section]
\newtheorem{lem}{Lemma}[section]

\newtheorem{prop}[lem]{Proposition}
\theoremstyle{definition}

\theoremstyle{remark}
\newtheorem{rem}{Remark}[section]
\numberwithin{equation}{section}

\newcommand{\set}[1]{\left\{#1\right\}}

\newcommand{\calC}{\mathcal{C}}

\newcommand{\calU}{\mathcal{U}}

\newcommand{\bbZ}{\mathbb{Z}}

\newcommand{\bbQ}{\mathbb{Q}}

\newcommand{\bbC}{\mathbb C}

\newcommand{\bbN}{\mathbb N}
\newcommand{\bbT}{\mathbb{T}}

\newcommand{\GL}{ \mbox{GL}}

\newcommand{\Sp}{\mathrm{Sp}}

\newcommand{\calH}{\mathcal{H}}

\begin{document}
\title[Scarring with simple spectrum]
{Scarring for Quantum Maps with Simple Spectrum}%
\author{Dubi Kelmer}%
\address{School of Mathematics,
Institute for Advanced Study, 1 Einstein Drive , Princeton, New Jersey 08540 US}
\email{kelmerdu@ias.edu}

\thanks{
 This material is based upon work supported by the National Science Foundation under agreement No. DMS-0635607. Any opinions, findings and conclusions or recommendations expressed in this material are those of the author and do not necessarily reflect the views of the National Science Foundation.}%

\date{\today}%
\dedicatory{}%
\commby{}%

\begin{abstract}
In \cite{Kelmer06Scar} we introduced a family of symplectic maps of the torus whose quantization exhibits scarring on invariant co-isotropic submanifolds. The purpose of this note is to show that in contrast to other examples, where failure of Quantum Unique Ergodicity is attributed to high multiplicities in the spectrum, for these examples the spectrum is (generically) simple.
\end{abstract}

\maketitle
\section{Introduction}
A significant problem in quantum chaos is to understand the behavior of eigenstates of classically chaotic
systems in the semiclassical limit. In particular, one would like
to classify the possible measures on phase space obtained as a
quantum limit.
For surfaces of negative
curvature it is conjectured that the only possible limiting measure is the volume measure \cite{RudSarnak94}.
This is usually referred to as Quantum Unique Ergodicity (QUE).
The QUE conjecture has been proved for the case of compact arithmetic
surfaces if one takes into account the arithmetic symmetries of the system \cite{Linden06}.
If we assume that the spectral multiplicities are bounded (which is believed to be true in this case), then this would imply this result for any basis of eigenfunctions on these arithmetic surfaces.

A similar situation also occurs for the quantized cat map.
Here again, the quantized system exhibits arithmetic
symmetries, and after taking these symmetries into account the
only possible limiting measure is shown to be the volume
measure \cite{KurRud2000}. However, here the multiplicities in the spectrum are far from being bounded,
so we can not deduce it for arbitrary eigenfunctions. Indeed, there is an explicit construction of a
thin sequence of eigenstates of
the propagator (without the symmetries) that become partially localized on periodic orbits
\cite{FaureNonnenmacherDebievre03}. Another violation of QUE is demonstrated for the Walsh quantized baker map on $\bbT^2$ \cite{AnantharamanNonnenmacher07}. Here there are limiting quantum measures supported on fractal subsets of the torus.
In this example, again, the quantum propagator has very large degeneracies in the spectrum.

In \cite{Kelmer05Arith} another instance of scarring is presented for the quantization of certain linear maps of a multidimensional torus.
In this construction there are limiting measures that localize on a co-isotropic invariant manifold (rather then on a periodic orbit).
The spectrum of the propagator in these examples also has large multiplicities, however, in contrust to the cat map scars in \cite{FaureNonnenmacherDebievre03}, this latter type of scarring still occurs when taking the arithmetic symmetries into account.
This might suggest that the scarring on co-isotropic manifolds is not due to the spectral degeneracies.
Moreover, in \cite{Kelmer06Scar} it is shown that this phenomenon also occurs for certain nonlinear perturbations of these linear map.
Since the parameter space for the perturbations is of infinite dimension, it is reasonable to expect that the perturbation would kill all
of the spectral degeneracies.

In order to clarify this point, we will show here (on the simplest example of the construction in \cite{Kelmer06Scar}) that this is indeed the case. That is, that there are nonlinear symplectic maps on $\bbT^4$ that exhibit scarring on invariant manifolds and yet have a simple spectrum.
We will consider maps of the form
\[\Phi_g(\begin{pmatrix}p\\q\end{pmatrix})= \begin{pmatrix} B^t p+\nabla g(B^{-1} q)\\ B^{-1}q\end{pmatrix},\]
where $B\in\GL(2,\bbZ)$ is a fixed hyperbolic matrix and $g\in C^\infty(\bbT^2)$ a smooth real valued function.
For any choice of $g$ (sufficiently small) this map is a symplectic map of Anosov type and hence a nice model for chaotic dynamics. Its quantization is a family of unitary operators, for each $N\in\bbN$ an operator $U_N(\Phi_g)$ acting on finite dimensional Hilbert space $\calH_N$ of dimension $N^2$. The manifold $X_0=\set{\begin{pmatrix}p\\0\end{pmatrix}}\subset\bbT^4$ is invariant under the map $\Phi_g$, and Lebesgue's measure on $X_0$ is obtained as a limiting quantum measure. Nevertheless, we show that the propagator $U_N(\Phi_g)$ (generically) has a simple spectrum.
\begin{thm}\label{t:main}
There is an open and dense subset $\calC\subseteq C^\infty(\bbT^2)$ such that for any $g\in\calC$ (and any $N$) the operator
$U_N(\Phi_g)$ has a simple spectrum.
\end{thm}

\begin{rem}
We note that the maps $\Phi_g$ with $g\in\calC$ can be considered generic only within the (small) family of perturbations $g\in C^\infty(\bbT^2)$.
If we consider the space of perturbation by all Hamiltonians $g\in C^\infty(\bbT^4)$ (not to mention the space of all symplectic maps), then these maps are in no way generic. (For generic maps, in the latter sense, we would not expect there to be scarring.) Hence, this result should not be viewed as a claim on generic quantum maps, but rather as the claim that there exists a quantum map with a simple spectrum that violates QUE.
\end{rem}

% \begin{rem}
% A weaker condition on the possible limiting measure for chaotic systems is a lower bound on their entropy \cite{AnantharamanNonnenmacher06,AnantharamanNonnenmacher07}.
% In particular, for geodesic flow on a manifold with constant negative curvature the entropy of any limiting measure is shown to be at least half the maximal possible entropy \cite{AnantharamanNonnenmacher06} (note that the entropy for the volume measure is the maximal entropy so this is weaker then the QUE conjecture) . For more general systems (e.g., for variable curvature) weaker bounds are known, however, it is conjectured that half the maximal entropy is the correct lower bound.
% \end{rem}
\subsection*{Acknowledgment:}
I wish to thank Peter Sarnak and Lior Silberman for our discussions.

\section{Background}
\subsection{Classical dynamics}
The classical dynamics is given by the iteration of the folowing symplectic map on $\bbT^4$:
Let $B\in\GL(2,\bbZ)$ be hyperbolic with $\det(B)=-1$ (e.g., take $B=\begin{pmatrix} 1 & 2\\ 2& 3\end{pmatrix}$). Let $g\in C^{\infty}(\bbT^2)$ be a smooth real valued function on $\bbT^2$.
We will use the coordinates $\begin{pmatrix}p\\q\end{pmatrix}\in\bbT^4$ with $p=\begin{pmatrix}p_1\\p_2\end{pmatrix}\in\bbT^2$ and $q=\begin{pmatrix}q_1\\q_2\end{pmatrix}\in\bbT^2$.
We then consider the map 
\[\Phi_g(\begin{pmatrix}p\\q\end{pmatrix})= \begin{pmatrix} B^tp+\nabla g(B^{-1} q)\\ B^{-1}q\end{pmatrix}\]
where $B^t$ is the transpose of $B$ and $\nabla g=\begin{pmatrix}\frac{\partial g}{\partial q_1}\\ \frac{\partial g}{\partial q_2}\end{pmatrix}$ is the gradient of $g$.

This map is the composition of the linear map $A=\begin{pmatrix} B^t & 0\\ 0 & B^{-1}\end{pmatrix}$ and the Hamiltonian flow $(\begin{pmatrix}p\\q\end{pmatrix})\mapsto \begin{pmatrix} p+T\nabla g(q)\\ q\end{pmatrix}$ evaluated at time $T=1$, (we think of $g$ as a Hamiltonian function on $\bbT^4$ that depends only the $q$ coordinates). Since we assume $B$ is hyperbolic then so is $A$, hence, by structural stability (assuming the Hamiltonian $g$ is sufficiently small) the map $\Phi_g$ is Anosov.

\subsection{Quantum dynamics}
The procedure for quantizing linear maps and Hamiltonian flows (or maps) on the torus is described in \cite{Kelmer06Scar}.
We briefly review this procedure specifically for the maps $\Phi_g$ defined above. 
The Hilbert space of quantum states is $\calH_N=L^2[(\bbZ/N\bbZ)^2]$,
and the semiclassical limit is achieved by taking $N\rightarrow \infty$. 
% For $f\in C^\infty(\bbT^{4})$ a smooth observable, we denote by $\Op_N(f):\calH_N\rightarrow \calH_N$ its quantization. 
% For the general definition we refere to \cite{Kelmer06Scar}, we note here that for $f$ a function only of position (i.e., the $q$ coordinates) the operator $\Op_N(f)$ is simply given by
% \[\Op(f)\psi(Q)=f(\frac{Q}{N})\psi(Q).\]
The quantization of the map $\Phi_g$ is a unitary operator $\calU_N(\Phi_g)$ acting on $\calH_N$. For a general symplectic linear map, $A\in\Sp(4,\bbZ)$, and a Hamiltonian, $g\in C^\infty(\bbT^4)$, the definition of this operator is more complicated, however, for our specific choice of matrix of the form
$A=\begin{pmatrix} B^t & 0\\ 0 & B^{-1}\end{pmatrix}$ and a Hamiltonian that depends only on the position coordinates the propagator is given by the following (simple) formula:
\[U_N(\Phi_g)\psi(Q)=e(Ng(\frac{Q}{N}))\psi(BQ).\]
The map $\Phi_g$ is a special case of the maps considered in \cite{Kelmer06Scar} and in particular it exhibit scarring on an invariant manifold. In fact, for this map it is easy to see that the function $\psi_0=N\delta_0 \in\calH_N$ is an eigenfunction of $U_N(\Phi_g)$ with corresponding Wigner distribution being Lebesgue's measure on $X_0$.
\section{Proof of Theorem \ref{t:main}}
\subsection{Spectrum}
We can use the simple formula for the quantum propagator $U_N(\Phi_g)$ in order to explicitly compute the spectrum.
\begin{prop}\label{p:eigenvalues}
The eigenvalues of $U_N(\Phi_g)$ are given by the folowing formula:
\[\lambda_{\xi,j}=e(\frac{N}{T_\xi}\sum_{k=1}^{T_\xi}g(\frac{B^k\xi}{N})+\frac{j}{T_\xi})),\]
where $\xi\in(\bbZ/N\bbZ)^2$ runs through representatives of $B$-orbits, $T_\xi$ is the size of the corresponding orbit and for each $\xi$ we take $1\leq j\leq T_\xi$.
\end{prop}
\begin{proof}
Let $\psi$ be an eigenfunction with eigenvalue $\lambda$, then we have for all $Q\in(\bbZ/N\bbZ)^2$
\[U_N(\Phi_g)\psi(Q)=e(Ng(\frac{Q}{N}))\psi(BQ)=\lambda \psi(Q).\]
or equivalently
\[\psi(BQ)=\lambda e(-Ng(\frac{Q}{N}))\psi(Q).\]
Iterating this $k$ times we get
\[\psi(B^kQ)=\lambda^k e(-N\sum_{j=0}^{k-1}g(\frac{B^jQ}{N}))\psi(Q).\]

For any $\xi\in(\bbZ/N\bbZ)^2$ let $[\xi]\subset (\bbZ/N\bbZ)^2$ denote its orbit under $B$ and let $T_\xi=\sharp [\xi]$ be its order.
Then for any orbit $[\xi]$ either $\psi(Q)=0$ for all $Q\in[\xi]$ or 
$\lambda^{T_\xi}=e(N\sum_{Q\in[\xi]}g(\frac{Q}{N}))$.
We thus see that the only possible eigenvalues are of the form $\lambda_{\xi,j}$ as above.

It now remains to show that for any orbit $[\xi]$ and $1\leq j\leq T_\xi$ there is a corresponding eigenfunction with eigenvalue $\lambda_{\xi,j}$
(since there are $N^2$ such pairs this would imply that these are all the eigenvalues).
For any pair $([\xi],j)$ define a function $\psi_{\xi,j}$ supported on $[\xi]$ and defined there by
\[\psi_{\xi,j}(B^k\xi)=\lambda_{\xi,j}^ke(-N\sum_{j=0}^{k-1}g(\frac{B^j \xi}{N})).\]
Then for any $Q=B^k\xi\in [\xi]$ we have
\[\psi_{\xi,j}(BQ)=\psi_{\xi,j}(B^{k+1}\xi)=\lambda_{\xi,j}e(-Ng(\frac{Q}{N}))\psi_{\xi,j}(Q),\]
so that indeed $\psi_{\xi,j}$ is an eigenfunction of $U_N(\Phi_g)$ with eigenvalue $\lambda_{\xi,j}$.
\end{proof}
\subsection{Multiplicities}
In order to show that there are no multiplicities we need to show that (for any value of $N$) the $N^2$ eigenvalues $\lambda_{\xi,j}$ are all different. In order to insure this we will take the function $g$ from the folowing (dense) subset of functions:\\
For any $y\in \bbC^{n}$, any $n$-tuple $x=(x_1,\ldots,x_n)$ of distinct points in $\bbT^2$ and any $z\in \bbC$ let
\[\calC(n,y,x,z)=\set{g\in C^\infty(\bbT^2)|\sum_{j=1}^ny_jg(x_j)\neq z}.\]
For any given $(n,y,x,z)$, the set $\calC(n,y,x,z)$ is open and dense in $C^\infty(\bbT^2)$.
Consequently (by Baire's Category Theorem) so is the countable intersection
\[\calC=\bigcap_{n=1}^\infty \bigcap_{y\in \bbZ^n}\bigcap_{z\in \bbQ}\bigcap_{x\in (\bbQ^2/\bbZ^2)^n}\calC(n, y, x,z),\]
(where the last intersection is over tuples of $n$ distinct points in $\bbQ^2/\bbZ^2$). The proof of Theorem \ref{t:main} is now 
concluded with the folowing lemma.
\begin{lem}\label{p:simple}
For any $g\in \calC$ we have that $\lambda_{\xi,j}=\lambda_{\eta,k}$ if and only if $[\xi]=[\eta]$ and $j=k$.
\end{lem}
\begin{proof}
Assume $\lambda_{\xi,j}=\lambda_{\eta,k}$ then
\[e(\frac{N}{T_\xi}\sum_{n=1}^{T_\xi}g(\frac{B^n\xi}{N})+\frac{j}{T_\xi}-\frac{N}{T_\eta}\sum_{m=1}^{T_\eta}g(\frac{B^m\eta}{N})-\frac{k}{T_\eta})=1.\]
If $[\xi]=[\eta]$ then $e(\frac{j-k}{T_\xi})=1$ implying that $j=k$ as well. Otherwise, the points
$$\{\frac{B^{n}\xi}{N},\frac{B^{m}\eta}{N}|1\leq n\leq T_\xi,\;1\leq m \leq T_\eta\}$$ are distinct rational points satisfying
\[\sum_{n=1}^{T_\xi}T_\eta g(\frac{B^n\xi}{N})-\sum_{m=1}^{T_\eta}T_\xi g(\frac{B^m\eta}{N})\in\bbQ,\]
in contradiction to the assumption that $g\in\calC$.
\end{proof}

%-----------------------------------------------------------------
\def\cprime{$'$}
\providecommand{\bysame}{\leavevmode\hbox to3em{\hrulefill}\thinspace}
\providecommand{\MR}{\relax\ifhmode\unskip\space\fi MR }
\providecommand{\MRhref}[2]{%
  \href{http://www.ams.org/mathscinet-getitem?mr=#1}{#2}
}
\providecommand{\href}[2]{#2}

% ----------------------------------------------------------------
%GATHER{/home/member/kelmerdu/My_Documents/Tex/Bib/Mybib.bib}   % For Gather Purpose Only
%\bibliographystyle{amsplain}
%\bibliography{/home/member/kelmerdu/My_Documents/Tex/Bib/Mybib.bib}
\end{document}